\documentclass[sn-mathphys-num]{sn-jnl}

\usepackage{graphicx}%
\usepackage{multirow}%
\usepackage{amsmath,amssymb,amsfonts}%
\usepackage{amsthm}%
\usepackage{amsmath, bm}
\usepackage{mathrsfs}%
\usepackage[title]{appendix}%
\usepackage{xcolor}%
\usepackage{textcomp}%
\usepackage{manyfoot}%
\usepackage{booktabs}%
\usepackage{algorithm}%
\usepackage{algorithmicx}%
\usepackage{algpseudocode}%
\usepackage{listings}%

\theoremstyle{thmstyleone}%

\theoremstyle{thmstyletwo}%

\theoremstyle{thmstylethree}%

\newtheorem{lem}{Lemma}
\newtheorem{thm}{Theorem}

\newtheorem{claim}{Claim}

\newtheorem{defi}{Definition}

\newtheorem{Proposition}[lem]{Proposition}

\begin{document}

\title[Article Title]{An extension of Dembo-Hammer's reduction algorithm for the 0-1 knapsack problem}

\author{Yang Yang}\email{zhugemutian@outlook.com}

\affil{\orgdiv{School of Mathematical Sciences}, \orgname{Xiamen University}, \orgaddress{\street{Siming South Road 422\#}, \city{Xiamen}, \postcode{361005}, \state{Fujian}, \country{PR China}}}

\abstract{Dembo-Hammer's Reduction Algorithm (DHR) is one of the classical algorithms for the 0-1 Knapsack Problem (0-1 KP) and its variants, which reduces an instance of the 0-1 KP to a sub-instance of smaller size with reduction time complexity $O(n)$. We present an extension of DHR (abbreviated as EDHR), which reduces an instance of 0-1 KP to at most $n^i$ sub-instances for any positive integer $i$. In practice, $i$ can be set as needed. In particular, if we choose $i=1$ then EDHR is exactly DHR. Finally, computational experiments on randomly generated data instances demonstrate that EDHR substantially reduces the search tree size compared to CPLEX.}




\maketitle

\section{Introduction}

Given an item set $N=\{1,2,\cdots,n\}$, let $p_j$ and $w_j$ denote the profit and weight of the $j$-th item, respectively. In classical 0-1 Knapsack Problem (0-1 KP), the goal is to select a subset of items from the set $N$ such that the sum of their weights does not exceed a given capacity $C$. The objective is to maximize the total profits of the chosen items\cite{Kellerer2004, Pisinger1997, Balas1980, Bellman1957, Martello1999, Martello2000}. In terms of $(0,1)$-vector, the 0-1 KP can be formulated as the following programming:

\begin{defi}\label{0-1KP} (0-1 KP).
    \begin{align}\label{KP_model-1}
        \max f(\bm{X})=\sum\limits_{j=1}^n x_jp_j
    \end{align}
subject to
    \begin{align}
        & g(\bm{X})=\sum\limits_{j=1}^n x_jw_j\le C\label{KP_model-2}\\
        & \hspace{35pt} x_j\in \{0,1\}\label{KP_model-3}
    \end{align}
\end{defi}

For $j \in N$, $x_j=1$ indicates that the $j$-th item is packed in the knapsack while $x_j=0$ indicates not. For simplicity, we assume that $p_j$, $w_j$ and $C$ are positive integers for any $j\in N$ \cite{Kellerer2004}. Meanwhile, in order to avoid trivial solution, we assume $w_j<C$ for any $j\in N$ and $\sum_{j=1}^n w_j>C$. In the algorithms of the 0-1 KP typically employed, one of the key points for solving the 0-1 KP is initially to order the variables according to non-increasing profit-to-weight $e_j=p_j/w_j$ ratios, also called the {\it profit density}. Therefore, we also assume the following
$$e_1\ge e_2\ge\cdots\ge e_n.$$

The 0-1 KP is known as NP-hard problem \cite{Garey1979,Karp1972}. Except the Dynamic Programming \cite{Bellman1957} that can exactly solve the 0-1 KP in pseudo-polynomial time, there is currently no polynomial time complexity algorithm that can exactly solve the 0-1 KP. Therefore, various methods or strategies for fast dimensionality reduction of problems in polynomial time have received much attention, that is, to reduce the size of an instance of the 0-1 KP through a reduction algorithm with polynomial complexity that partitions the item set $N$ into three subsets $N_0,N_1$ and $F$ so that the items in $N_1$ are all included in any optimal solution while every item in $N_0$ is not. Thus the optimal solution of the instance is given by $N_1+N_F$, where $N_F$ is an optimal solution of the sub-instance of the original one restricted on $F$, that is, an instance of size $|F|$.

Along this direction, a number of reduction algorithms were proposed. For examples, we refer to the Ingargiola and Korsh's Reduction algorithm (IKR) \cite{ Ingargiola1973} with time complexity of $O(n^2)$ by Dantzig bound \cite{ Dantzig1957};  Martello and Toth's Reduction algorithm(MTR) with time complexity to $O(n\log n)$ by Reduction with Complete Sorting(RCS) and Reduction with Partial Sorting(RPS) \cite{Martello1988}. Further, based on MTR, in 1990 Martello and Toth  proposed MTR2 \cite{Martello1990} to get a better solution.

In addition, Dembo and Hammer proposed a Reduction algorithm (DHR) \cite{ Dembo1980}, which reduces an instance of the 0-1 KP with $n$ items to be a sub-instance of
$$|F|=n-\left|\left\{j< b: \frac{p_j}{w_j+r}>\frac{p_b}{w_b}\right\}\right|-\left|\left\{j\ge b: \frac{p_j}{w_j-r}<\frac{p_b}{w_b}\right\}\right|$$
items with reduction time complexity  $O(n)$, where $r$ is the residual capacity, i.e.,  $r=C-\sum_{k=1}^{b-1}w_k\notag$, and $b$ is the break item( also called the $split\ item$ in literature \cite{Pisinger2005}), i.e. $b = \min\{k:\sum_{j=1}^kw_j>C\}$.  DHR has received widespread attention because of its simplicity and effectiveness, and its ease of hybridizing with other algorithms. Although DHR alone is not as efficient as IKR, MTR and MTR2\cite{ Martello1988,Fahle2002},  Pisinger  in 1995 presented EXPKNAP \cite{ Pisinger1995a} based on the core strategy \cite{ Balas1980}, which has better performance than MTR and MTR2. Later in 1997, MINKNAP \cite{ Pisinger1997} was proposed based on EXPKNAP and DHR, the performance of which is better than EXPKNAP.

In addition to being used to solve the 0-1 KP, DHR also has more applications for some extended models of the knapsack problem. Tsesmetzis et al. \cite{ Tsesmetzis2008a} transformed QoS-aware problem to Selective Multiple Choice Knapsack Problem and  designed an algorithm with time complexity between $O(n\log n)$ and $O(n^2)$ through DHR as lower bound, which increases the provider's profit up to 0.5\% on average. Egeblad and Pisinger solved the two- and three-dimensional knapsack packing problem with semi-normalized packing algorithm and DHR \cite{ Pisinger2009a}. Using DHR, Pisinger and Saidi also  analysed the tolerance of 0-1 KP \cite{ Pisinger2016a}.

Recently, Dey et al.\cite{ Dey2023} proposed a method to analyse the upper bound of the nodes in search tree of the Branch and Bound algorithm, and prove that the branch and bound algorithm can solve random binary integer programming in polynomial time.

In this paper, we propose an extension of Dembo-Hammer's Reduction Algorithm (EDHR). For any positive integer $i$, the algorithm EDHR reduces an instance of KP with $n$ items to be $n^{i}$ sub-KP sub-instances of
$$|F|=n-\left|\left\{j< b: \frac{p_j}{w_j+r/i}>\frac{p_b}{w_b}\right\}\right|-\left|\left\{j\ge b: \frac{p_j}{w_j-r/i}<\frac{p_b}{w_b}\right\}\right|.$$
items with reduction time complexity $O(n)$.

In practice, $i$ can be set by need. In particular, if we choose $i=1$ then EDHR is exactly DHR.  Finally, we perform the computational experiment for some data instances that are constructed randomly. Our experiment shows that, compared with CPLEX, EDHR significantly decreases the search tree size for the instances. Our method also reduces the interval gap of the distances from power of 2 to integer and decreases the complexity of the method given by Dey et al.

\section{Dembo and Hammer's Reduction Algorithm}

If we relax the integrality constraint $x_j\in \{0,1\}$ to the linear constraint $0\leq x_j\leq 1$, we obtain the Linear Knapsack Problem (LKP) \cite{Pisinger1997}. Let  $\boldsymbol{X}^*=(x^*_1,x^*_2,\ldots,x^*_n)$ be an optimal solution to LKP, where $0\leq x^*_j\leq 1$ for each $j\in\{1,2,\ldots,n\}$. It is clear that $x_j^*=1$ if $j<b$, $x_j^*=0$ if $j>b$ and $x^*_b=(C-\sum_{i=1}^{b-1}w_i)/w_b$. This yields naturally an upper bound, called Dantzig bound, for the 0-1 KP\cite{ Dantzig1957}:
\begin{align}
    U=\sum\limits_{k=1}^{b-1}p_k+\left\lfloor rp_b/w_b\right\rfloor \notag,
\end{align}
where $\lfloor x\rfloor$ is the greatest integer no more than $x$ and $r=C-\sum\limits_{k=1}^{b-1}w_k\notag$, called the  the {\it residual capacity}.

On the other hand, the integer solution $\boldsymbol{X}'=(x^*_1,\ldots,x^*_{b-1},0,\ldots,0)$ is a solution to KP, which is known as the
{\it break solution}. This yields naturally a lower bound of the 0-1 KP\cite{ Pisinger1997, Pisinger1995a}, i.e.,
\begin{align}
    L= \sum\limits_{k=1}^{b-1}p_k\notag.
\end{align}

Let $\boldsymbol{X}=(x_1,x_2,\ldots,x_n)$ be an arbitrary solution of KP. Note that the upper bound $U$ and lower bound $L$ do not mean that  $x_j = 1$ for every $j= 1,2,\ldots, b-1$ and $x_j=0$ for $j=b,\ldots, n$. Moreover, the items where $x_j\not=x'_j$ are generally very close to the break item $b$. Pisinger attempted to test this conclusion by constructing 1000 data instances with $n = 1000$, where $p_j$ and $w_j$ were randomly distributed within the interval $[1, 1000]$. The capacity $C$ was chosen such that the break item $b$ was set to 500 for all instances. Items in each data instance are ordered according to non-increasing profit density. The computational experiment described in \cite{Pisinger1995a} revealed that, on average, there were only about 3.4 such items per instance with $n = 1000$. Theoretically, Dembo and Hammer proved the following result.

Pisinger attempted to test this conclusion by constructing 1000 data instances with $n = 1000$, where $p_j$ and $w_j$ were randomly distributed within the interval $[1, 1000]$. The capacity $C$ was chosen such that the break item $b$ was set to 500 for all instances. Items in each data instance are ordered according to non-increasing profit density. The computational experiment described in \cite{Pisinger1995a} revealed that, on average, there were only about 3.4 such items per instance with $n = 1000$.

\begin{thm}\label{DHR} \cite{ Dembo1980, Pisinger1995a} Let $\bm{Y}=(y_1,y_2,\ldots,y_n)$ be the optimal solution. For any $j=1,\cdots,b-1$, if
    \begin{equation}\label{DHR-1}
        \left|\begin{array}{cc}
        p_j & r+w_j\\
        p_b & w_b
        \end{array}\right|>0,
    \end{equation}
    then $y_j=1$, that is, the item $j$ is included in the optimal solution.

  Further, for any $j=b,\cdots,n$, if
    \begin{equation}\label{DHR-2}
        \left|\begin{array}{cc}
        -p_j & r-w_j\\
        p_b & w_b
        \end{array}\right|>0,
    \end{equation}
  then  $y_j=0$, that is, the item $j$ is not included in the  optimal solution.
\end{thm}

Let $N_{1,1}$ denote the set of items in $\{1,\ldots,b-1\}$ that satisfy inequality (\ref{DHR-1}), and $N_{1,4}$ the set of items in $\{b,\ldots,n\}$ that satisfy inequality (\ref{DHR-2}).

According to Theorem 1, every item in $N_{1,1}$ is included in any optimal solution and, in contrast, no item in $N_{1,4}$ is included in an optimal solution. Thus, the original KP could be reduced to be a sub-KP $F$ of $n-|N_{1,1}|-|N_{1,4}|$ items and capacity $C_F=C-\sum_{i\in N_{1,1}}w_i$.

\section{Main result}

In this section, we give an extension of DHR Algorithm. The main idea is to extend the size of $N_{1,1}$ and $N_{1,4}$ determined by the DHR Algorithm.

Let $j,k$ $(1\le j,k\le b-1)$ be two items such that none of them satisfies (\ref{DHR-1}). Let $K^*$ be the instance obtained from the original problem by combining the items $j$ and $k$ to be a new item $t$ with profit $p^*_t=p_j+p_k$ and weight $w^*_t=w_j+w_k$. Moreover, we assume that the items in $K^*$ are ordered according to non-increasing profit density.  Since $p_j/w_j\geq p_b/w_b$ and  $p_k/w_k\geq p_b/w_b$, we have  $p^*_t/w^*_t\geq p_b/w_b$. Moreover, it is clear that the break item of $K^*$ is exactly that of $K$ and, therefore, $p^*_{b-1}=p_b$ and $w^*_{b-1}=w_b$. Let $\bm{X}''$ be an optimal solution  of $K^*$. If
\begin{align}\label{EDHR}
    \frac{p_j+p_k}{w_j+w_k+r}>\frac{p_b}{w_b},
\end{align}
i.e., $p^*_t/(w^*_t+r)>p^*_{b-1}/w^*_{b-1}$, then by inequality (\ref{DHR-1}), the item $t$ must be included in $\bm{X}''$.

To facilitate further discussion, let $\bm{Y}$ represent the optimal solution, where $y_j = 1$ indicates that the $j$-th item is selected by the optimal solution $\bm{Y}$, and $y_j \neq 1$ indicates that it is not selected.

\begin{Proposition} If the inequality (\ref{EDHR}) is satisfied, then any  optimal solution  $\bm{Y}$ of the 0-1 KP contains at least one of the two items $j$ and $k$. Equivalently, at most one of the two items $j$ and $k$ is not included in the optimal solution $\bm{Y}$.
\end{Proposition}
\begin{proof}Suppose to the contrary that neither $j$ nor $k$ is included in  $\bm{Y}$. By the definition of $K^*$, we have $f(\bm{Y})\geq f(\bm{X}'')$. This means  $\bm{Y}$ restricted on $N\setminus\{j,k\}$ is a feasible solution of $K^*$. This is a contradiction and the claim follows.
\end{proof}

For the $j$-th item, if $e_j > e_b$, satisfies inequality (\ref{EDHR}), but does not satisfy inequality (\ref{DHR-1}), then the optimal solution $\bm{Y}$ maybe not include the $j$-th item, i.e., $y_j = 0$. If there are two items $j$ and $k$ such that $1 \leq j, k \leq b-1$ and they satisfy inequality (\ref{EDHR}) with $x_j = x_k = 0$, a subproblem is generated. Moreover, we can derive an upper bound for this subproblem using Dantzig's bound, which is clearly less than the objective value of the break solution $\bm{X}'$. Consequently, the optimal solution must select at least one item between $j$ and $k$.

Moreover, if a pair of items $(j, k)$ that satisfy inequality (\ref{EDHR}) are collected as a set, we denote this set by $(j, k) \in N'_{1,1}$. Given that the computation results from CPLEX are used as the baseline in this paper, if a constraint is added to any pair of items that satisfy inequality (\ref{EDHR}), then at least $|N'_{1,1}|$ constraints of the form
\begin{align*}
x_j+x_k\ge 1
\end{align*}
should be added, for all pairs $(j, k) \in N'_{1,1}$. Obviously, this would result in a very large number of constraints, which significantly slow down the computation speed of CPLEX. Therefore, it is crucial for the new algorithm to consider whether inequality (\ref{EDHR}) can be effectively characterized by only a few constraints, or even a single constraint.

Notice that, if $p_j=p_k$ and $w_j=w_k$, then inequality (\ref{EDHR}) can be rewritten as
 \begin{align}\label{r2}
 \frac{p_j}{w_j+r/2}>\frac{p_b}{w_b}.
 \end{align}
Therefore, the above Proposition means that the set of the items $j$ that satisfies $p_j/w_j\geq p_b/w_b$ and inequality (\ref{r2}) contains at most one items that is not in the optimal solution. By applying inequality (\ref{r2}), we can represent the numerous constraints in inequality (\ref{EDHR}) with a single constraint. Generally, this motivates us to consider the set of the items $j$ that satisfy $p_j\geq p_b$ and
    \begin{align}\label{ri}
        \frac{p_j}{w_j+r/i}>\frac{p_b}{w_b}
    \end{align}
for any given integer $i\geq 2$. We will show in the following Theorem \ref{Ni1} that if the inequality (\ref{ri}) is satisfied, then the set has at most $i-1$ items that are not in the optimal solution.
\begin{defi}\label{Ni}
For any integer $i$ where $i\ge 1$, let $N=N_{i,1}\cup N_{i,2}\cup N_{i,3}\cup N_{i,4}\cup N_{i,5}$ be the partition of $N$, where
    \begin{align}
        & N_{i,1}=\{j: e_j>e_b, p_jw_b-p_b(r/i+w_j)>0\}\notag,\\
        & N_{i,2}=\{j:e_j>e_b, p_jw_b-p_b(r/i+w_j)\le 0\}\notag,\\
        & N_{i,3}=\{j:e_j\le e_b, w_j>r/i, p_jw_b+p_b(r/i-w_j)\ge 0\}\notag,\\
        & N_{i,4}=\{j:e_j\le e_b, w_j>r/i, p_jw_b+p_b(r/i-w_j)< 0\}\notag,\\
        & N_{i,5}=\{j:e_j\le e_b, w_j\le r/i\}\notag.
    \end{align}
And let
\begin{align}
        & F_{i,1}=\{j:j\in N_{i,2}\cup N_{i,3}, y_j=1, x^*_j=0\}\notag,\\
        & F_{i,2}=\{j:j\in N_{i,2}\cup N_{i,3}, y_j=x^*_j\}\notag,\\
        & F_{i,3}=\{j:j\in N_{i,2}\cup N_{i,3}, y_j=0, x^*_j=1\}\notag,\\
        & D_{i,1}=\{j:j\in N_{i,1}, y_j=0\}\notag,\\
        & D_{i,2}=\{j:j\in N_{i,4} , y_j=1\}\notag,\\
        & D_{i,3}=\{j:j\in N_{i,5}, y_j=1\}\notag.
    \end{align}
\end{defi}

Definition \ref{Ni} initially partitions the items with a profit density exceeding that of the break item into two sets: $N_{i,1}$ consists of items that satisfy inequality (\ref{ri}), while $N_{i,2}$ contains those that do not. Similarly, items with a value density less than or equal to the break item are categorized based on whether they satisfy the following inequality:
\begin{align}\label{ri2}
\frac{p_j}{w_j - r/i} > \frac{p_b}{w_b}.
\end{align}

Items satisfying this inequality are placed in set $N_{i,3}$, and those that do not in set $N_{i,4}$. Moreover, since the left side of inequality (\ref{ri2}) is always negative when $w_j < r/i$, it implies that these items are never selected by the optimal solution, which contradicts the actual scenario. Consequently, we require an additional set to describe these items, denoted by $N_{i,5}$.

Items in $N_{i,1}$ are typically selected by the break solution $\bm{X}'$ due to their high profit density. Therefore, only the number of items not selected in $N_{i,1}$ needs to be counted and denoted as $D_{i,1}$. Conversely, items in $N_{i,4}$ and $N_{i,5}$, because of their low profit density, are usually not selected. Thus, only the number of items selected in these sets needs to be statistically recorded, denoted as $D_{i,2}$ for $N_{i,4}$ and $D_{i,3}$ for $N_{i,5}$.

Items in $N_{i,2}$ and $N_{i,3}$ are likely to be selected in the break solution $\bm{X}'$, where $x^*_j = 1$ for $1 \leq j < b$, but not selected by the optimal solution $\bm{Y}$, or not selected by the break solution $\bm{X}'$ but chosen by the optimal solution $\bm{Y}$. In the optimal solution $\bm{Y}$, we will consider items from $N_{i,2}$ and $N_{i,3}$ that are not selected as $F_{i,1}$. Items from $N_{i,2}$ and $N_{i,3}$ that are selected in both the optimal solution $\bm{Y}$ and the break solution $\bm{X}'$ will be denoted as $F_{i,2}$. Items that are not selected in the optimal solution $\bm{Y}$ but are selected in the break solution $\bm{X}'$ will be denoted as $F_{i,3}$. According to Definition \ref{Ni}, we can derive the following result.

\begin{claim}\label{Di1}
If $|D_{i,1}|> i-1$, then
    \begin{align}
        \frac{\sum\limits_{j\in D_{i,2}}p_j+\sum\limits_{j\in D_{i,3}}p_j}{\sum\limits_{j\in D_{i,2}}w_j +\sum\limits_{j\in D_{i,3}}w_j}\le\frac{\sum\limits_{j\in D_{i,1}}p_j}{r+\sum\limits_{j\in D_{i,1}}w_j}.\notag
    \end{align}
\end{claim}
\begin{proof} Let $q\in D_{i,1}$ be such that $e_k\ge e_q$ for any item $k\in D_{i,1}$. Then we have
    \begin{align}
        \frac{\sum\limits_{j\in D_{i,2}}p_j+\sum\limits_{j\in D_{i,3}}p_j}{\sum\limits_{j\in D_{i,2}}w_j +\sum\limits_{j\in D_{i,3}}w_j}\le e_b\le \frac{p_q}{r/i+w_q} \le \frac{\sum\limits_{j\in D_{i,1}}p_j}{|D_{i,1}|\times r/i+\sum\limits_{j\in D_{i,1}}w_j} \le \frac{\sum\limits_{j\in D_{i,1}}p_j}{r+\sum\limits_{j\in D_{i,1}}w_j} .\notag
    \end{align}
\end{proof}

\begin{thm}\label{Ni1} For any positive integer $i$,  $N_{i,1}$ has at most $i-1$ items that are not in the optimal solution $\bm{Y}$.
\end{thm}

\begin{proof} Since $\boldsymbol{Y}$ is an  optimal solution, we have $f(\boldsymbol{Y})\ge f(\boldsymbol{X}')$. This means that
\begin{equation}\label{psum}
\sum\limits_{j\in D_{i,2}}p_j+\sum\limits_{j\in D_{i,3}}p_j\ge \sum\limits_{j\in D_{i,1}}p_j +\sum\limits_{j\in F_{i,3}}p_j-\sum\limits_{j\in F_{i,1}} p_j.
\end{equation}

On the other hand, we notice that $g(\boldsymbol{X}')+r = C$. Therefore, $g(\boldsymbol{Y})\le C= g(\boldsymbol{X}')+r$. Hence,
\begin{equation}\label{wsum} \sum\limits_{j\in D_{i,2}}w_j+\sum\limits_{j\in D_{i,3}}w_j\le \sum\limits_{j\in D_{i,1}}w_j +\sum\limits_{j\in F_{i,3}}w_j-\sum\limits_{j\in F_{i,1}} w_j+r.
\end{equation}

Suppose to the contrary that $|D_{i,1}|> i-1$. Then by (\ref{psum}), (\ref{wsum}) and Claim \ref{Di1}, we have
        \begin{align}
            &\sum\limits_{i\in D_{i,2}}p_j +\sum\limits_{j\in D_{i,3}}p_j=\left(\sum\limits_{j\in D_{i,2}}w_j +\sum\limits_{j\in D_{i,3}}w_j\right)\frac{\sum\limits_{j\in D_{i,2}}p_j +\sum\limits_{j\in D_{i,3}}p_j} {\sum\limits_{j\in D_{i,2}}w_j+\sum\limits_{j\in D_{i,3}}w_j}\notag\\
            \le & \left(\sum\limits_{j\in D_{i,1}}w_j+r\right) \frac{\sum\limits_{j\in D_{i,2}}p_j+\sum\limits_{j\in D_{i,3}}p_j}{\sum\limits_{j\in D_{i,2}}w_j +\sum\limits_{j\in D_{i,3}}w_j} +\left(\sum\limits_{j\in F_{i,3}}w_j -\sum\limits_{j\in F_{i,1}}w_j\right) \frac{\sum\limits_{j\in D_{i,2}}p_j +\sum\limits_{j\in D_{i,3}}p_j} {\sum\limits_{j\in D_{i,2}}w_j +\sum\limits_{j\in D_{i,3}}w_j}\notag\\
            <&\left(\sum\limits_{j\in D_{i,1}}w_j+r\right) \frac{\sum\limits_{j\in D_{i,1}}p_j} {\sum\limits_{j\in D_{i,1}}w_j+r} +\left(\sum\limits_{j\in F_{i,3}}w_j -\sum\limits_{j\in F_{i,1}}w_j\right)\frac{\sum\limits_{j\in D_{i,2}}p_j +\sum\limits_{j\in D_{i,3}}p_j} {\sum\limits_{j\in D_{i,2}}w_j +\sum\limits_{j\in D_{i,3}}w_j}\notag\\
            =&\sum\limits_{j\in D_{i,1}}p_j+\left(\sum\limits_{j\in F_{i,3}}w_j -\sum\limits_{j\in F_{i,1}}w_j\right) \frac{\sum\limits_{j\in D_{i,2}}p_j +\sum\limits_{j\in D_{i,3}}p_j} {\sum\limits_{j\in D_{i,2}}w_j +\sum\limits_{j\in D_{i,3}}w_j}
        \end{align}

    and
        \begin{align}\label{10}
            \sum\limits_{j\in F_{i,3}}p_j -\sum\limits_{j\in F_{i,1}} p_j &< \left(\sum\limits_{j\in F_{i,3}}w_j -\sum\limits_{j\in F_{i,1}}w_j\right) \frac{\sum\limits_{j\in D_{i,2}}p_j +\sum\limits_{j\in D_{i,3}}p_j} {\sum\limits_{j\in D_{i,2}}w_j +\sum\limits_{j\in D_{i,3}}w_j}.
        \end{align}

    We notice that the profit densities of the items in the set $F_{i,3}$ are more than that of the items in the set $F_{i,1}$. So by Definition \ref{Ni}, we have
$$ \sum\limits_{j\in F_{i,3}}p_j\sum\limits_{j\in F_{i,1}}w_j \ge\sum\limits_{j\in F_{i,3}}w_j\sum\limits_{j\in F_{i,1}}p_j,$$
where $\sum\limits_{j\in F_{i,3}}p_j$ is treated as zero if $F_{i,3}=\emptyset$. Hence,
$$ \sum\limits_{j\in F_{i,3}}p_j\sum\limits_{j\in F_{i,1}}w_j-\sum\limits_{j\in F_{i,1}}p_j\sum\limits_{j\in F_{i,1}}w_j\ge\sum\limits_{j\in F_{i,3}}w_j\sum\limits_{j\in F_{i,1}}p_j-\sum\limits_{j\in F_{i,1}}p_j \sum\limits_{j\in F_{i,1}}w_j,$$
i.e.,
\begin{equation}\label{11}
    \sum\limits_{j\in F_{i,3}}p_j-\sum\limits_{j\in F_{i,1}}p_j \ge \left(\sum\limits_{j\in F_{i,3}}w_j-\sum\limits_{j\in F_{i,1}}w_j\right)
    \frac{\sum\limits_{j\in F_{i,1}}p_j}{\sum\limits_{j\in F_{i,1}}w_j}.
\end{equation}
On the other hand, again by Definition \ref{Ni}, we have
\begin{equation}\label{12}
        \frac{\sum\limits_{j\in F_{i,1}}p_j}{\sum\limits_{j\in F_{i,1}}w_j}>\frac{\sum\limits_{j\in D_{i,2}}p_j +\sum\limits_{j\in D_{i,3}}p_j} {\sum\limits_{j\in D_{i,2}}w_j +\sum\limits_{j\in D_{i,3}}w_j}.
\end{equation}
Combining with inequalities (\ref{11}) and (\ref{12}), \begin{equation}\label{13}
    \sum\limits_{j\in F_{i,3}}p_j -\sum\limits_{j\in F_{i,1}} p_j > \left(\sum\limits_{j\in F_{i,3}}w_j -\sum\limits_{j\in F_{i,1}}w_j\right) \frac{\sum\limits_{j\in D_{i,2}}p_j +\sum\limits_{j\in D_{i,3}}p_j} {\sum\limits_{j\in D_{i,2}}w_j +\sum\limits_{j\in D_{i,3}}w_j}.
\end{equation}
  This is a contradiction, which completes the proof of Theorem \ref{Ni1}.
\end{proof}

By symmetry, the following results follows directly by a similar argument.

\begin{thm}\label{Ni4} For any positive integer $i$, $N_{i,4}$ has at most $i-1$ items that are in the optimal solution $\bm{Y}$.
\end{thm}

Let  $n_{i,1}=|N_{i,1}|,n_{i,4}=|N_{i,4}|$ and $N_{i,1}^* \subset N_{i,1}, N_{i,4}^* \subset N_{i,4}$. Then by Theorem \ref{Ni1}, we have $|N^*_{i,1}|\geq n_{i,1}-i+1$ and $|N^*_{i,4}|\leq i-1$. Notice that  $N_{i,1}$ has
$$\sum_{j=1}^{i}\binom{n_{i,1}}{n_{i,1}-j+1}\leq n_{i,1}^i$$
subsets of order at least $n_{i,1}-i+1$, and $N_{i,4}$ has
$$\sum_{j=1}^{i}\binom{n_{i,4}}{j-1}\leq n_{i,4}^i$$
subsets of order at most $i-1$.  Let $\bm{Y}_i^*$ be the optimal solution of the sub-instance of the original KP restricted on the subset $N_{i,2}\cup N_{i,3}\cup N_{i,5}$ and let
$${\cal Y}_i=\{\bm{Y}_i^*\cup N^*_{i,1}\cup N^*_{i,4}: N^*_{i,1}\subset N_{i,1}, |N^*_{i,1}|\ge n_{i,1}-i+1, N_{i,4}^*\subset N_{i,4}, N_{i,4}^*\le i-1\}.$$
Then by Theorem \ref{Ni1} and  Theorem \ref{Ni4}, it is clear that the optimal solution  $\bm{Y}$ of original problem is in ${\cal Y}_i$, i.e.,  $\bm{Y}\in {\cal Y}_i$. Further, we note that $|N_{i,2}\cup N_{i,3}\cup N_{i,5}|=n-n_{i,1}-n_{i,2}$ and $|{\cal Y}_i|\leq n_{i,1}^i n_{i,2}^i\leq n^{2i}$. This means that the original KP is reduced into at most $n^{2i}$ sub-instances of  $n-n_{i,1}-n_{i,2}$ items, the maximum optimal solution of which is precisely the optimal solution of the original KP. Based on Lemma 1, although the 0-1 KP is NP-hard, the decision variables of two subsets $N_{i,1}$ and $N_{i,4}$ can be exactly solved in time complexity $O(n^{2i})$.

In particular, the knapsack problem when all items have the same profit density is called the Sub-set problem(SSP) \cite{Karp1972}. We denote the problem that the number of items whose profit density are equal to the break item is finite as KP/SSP.

If the problem is KP/SSP and there is an integer $i$ such that all items whose profit density is more than the break item $b$ belong to the set $N_{i,1}$ and all items whose profit density is less than the break item $b$ belong to the set $N_{i,4}$, then the problem can be solved in time complexity $O(n^{2i})$ by EDHR.

Naturally, whether the constant $i$ has an upper bound becomes the key to solve the decision variables whose profit density are not equal to the profit density of the break item $b$ in polynomial time. In other words, if the constant $i$ has an upper bound, KP/SSP is $\mathcal{P}$.

\begin{thm}
    Constant $i$ has no upper bound.
\end{thm}
\begin{proof}
    Suppose to the contrary that constant $i$ has an upper bound, with no loss of generality, we let $m$ denote the upper bound of the constant $i$ and have
    \[
    \left|\begin{array}{ccc}
        p_j &    r/m+w_j\\
        p_b &    w_b
    \end{array}\right|>0
    \]
    for each item $j\in \{j|e_j>e_b,j\in N\}$.

    For the bound $m$, we let $e_b = \frac{p_b}{w_b}= \frac{2m}{2m+r}$, and the item profit and weight of an item $q\in N$ both be equal to 1. For any integer $i\le m$, if
    \[
    \left|\begin{array}{ccc}
        p_q &    r/i+w_q\\
        p_b &    w_b
    \end{array}\right|>0,
    \]
    then we have
    $$\frac{1} {r/i+1} =  \frac{p_q} {r/i+w_q} >\frac{p_b}{w_b}=\frac{2m}{2m+r},\notag$$
    and
    $$i>2m>m.\notag$$
    That is to say, for any value of $m$, we can always construct an instance, such that the value of constant $i$ is more than $m$. The proof of Theorem 4 is completed.
\end{proof}
Since constant $i$ has no upper bound, any item whose profit density is not equal to the profit density $e_b$ cannot be completely classified into $N_{i,1}$ and $N_{i,4}$, so the subproblem consist by $N_{i,2}$, $N_{i,3}$ and $N_{i,5}$ is still a NP-hard problem.

\section{Experimental results and comparative Analysis}

We perform all experimental computation on the device with Windows 11 Edition platform, Inter$\circledR$ Core$^{\text{TM}}$ i7-12700K CPU @ 3.60 GHz(20 CPUs), and 32 GB DDR5L of RAM(24 GB remaining). All Algorithms are computed by MATLAB 17a.

To effectively demonstrate the performance improvement, we use the computational results from CPLEX (version 12.8) as the baseline. Moreover, CPLEX (version 12.8) employs YALMIP \cite{yalmip} as an interface to call functions within MATLAB 17a.

\subsection{Benchmark instances and parameter setting}

In addition to the theoretical comparison, we also perform experimental comparison to verify the performance of the EDHR on large-scale problems.

In this section we conduct numerical experiments to compare the performance of EDHR with CPLEX based on five kinds of randomly generated instances, namely the Uncorrelated instances (UC), Weakly correlated instances (WC), Strongly correlated instances (SC) and Inverse strongly correlated instances (IC), Almost strongly correlated instances (ASC), respectively. Each kind has 10 data instances with different scale $n = 200, 400, 600, \dots,2000$ and same range $R = 1000$. For each data instance, the profit coefficients $p_j$ and weight coefficients $w_j$ are generated as follows \cite{Martello1999, Pisinger2005}:

1. UC instance: $p_j\in _z[1,R]$, $w_j\in _z[1,R]$, where $x\in _z[A,B]$ denotes $x$ is a random integer within interval $[A,B]$.

2. WC instance: $w_j\in _z[R/5+1,R]$, $p_j\in _z[w_j-R/5,w_j+R/5]$.

3. SC instance: $w_j\in _z[1,R]$, $p_j = w_j+R/5$.

4. IC instance: $p_j\in _z[1,R]$, $w_j = p_j+R/5$.

5. ASC instance: $w_j\in _z[1,R]$, $p_j\in _z[w_j+R/10-R/50, w_j+R/10+R/50]$.

For a better comparison, we set the break item by $b=\lfloor n/2\rfloor$, the weight of the break item $b$ by $w_b=\lfloor R/5\rfloor$, and the knapsack capacity by $C=\sum\limits_{j=1}^bw_j-1$. Thus, the residual capacity is $\lfloor R/5\rfloor-1$. Specifically, due to the unique characteristics of IC instances, if we set $w_b = R/5$, then $p_j = 0$, which renders the instances meaningless. Consequently, we let $w_b = 2 \times R/5$ in the IC data, resulting in $p_b = R/5$.

To effectively demonstrate the algorithm's performance improvement, we employ the results from CPLEX as our baseline. CPLEX is widely utilized by industrial researchers and is regarded as a standard baseline in operations research. It should be noted that MATLAB supports CPLEX up to version 12.10 and is no longer available for download from the official website, precluding access to higher versions for computation results. Due to these software limitations, we have chosen to use CPLEX version 12.8 for comparison. This version's computation results are frequently used for baseline, as seen in the literature \cite{Wei2019}. Consequently, the results from CPLEX (version 12.8) serve as an appropriate baseline for our comparison.

For the EDHR, the computational effect varies significantly with different values of $i$. Given that CPLEX is renowned for its rapid solving speed, if adding a large number of constraints yields only a marginal improvement in the performance, the solution speed may actually decrease. After careful consideration, we have decided to set the parameter $i$ to 2 within the EDHR.

Meanwhile, due to CPLEX's rapid computation speed, factors as computer response time can significantly affect the overall solving time more than the algorithm itself. A straightforward comparison of computation times is therefore susceptible to considerable error, even when computation times are averaged over multiple runs. For a given instance, CPLEX records the size of the search tree nodes, known as 'ticks,' throughout the solution process, which remains constant. Typically, a lower tick count correlates with a shorter solution time. Consequently, to better evaluate the algorithm's solving speed, we are employing the value of ticks as a metric.

\subsection{Computational results and comparisons}

\begin{table}[h]
\caption{Solving performances of CPLEX and EDHR on UC}
\centering
\scriptsize
\begin{tabular}{cccccccc}
\toprule
\multirow{2}{*}{instance} & \multirow{2}{*}{n} & \multicolumn{2}{c}{CPLEX} &  & \multicolumn{2}{c}{EDHR} & \multirow{2}{*}{rate} \\ \cmidrule{3-4} \cmidrule{6-7}
                          &                    & result       & ticks      &  & result      & ticks      &                       \\ \hline
UC01 &  200 &  62673 &  3.02 &  &  62673 &  1.08 & 64.24\% \\
UC02 &  400 & 135611 &  5.04 &  & 135611 &  2.73 & 45.83\% \\
UC03 &  600 & 205941 &  6.01 &  & 205941 &  2.60 & 56.74\% \\
UC04 &  800 & 267277 & 10.76 &  & 267277 &  3.61 & 66.45\% \\
UC05 & 1000 & 338991 &  9.73 &  & 338991 &  5.84 & 39.98\% \\
UC06 & 1200 & 398597 & 14.71 &  & 398597 &  7.06 & 52.01\% \\
UC07 & 1400 & 471846 & 15.84 &  & 471846 &  7.69 & 51.45\% \\
UC08 & 1600 & 532806 & 18.49 &  & 532806 &  7.74 & 58.14\% \\
UC09 & 1800 & 599534 & 21.65 &  & 599534 &  9.86 & 54.46\% \\
UC10 & 2000 & 672594 & 23.05 &  & 672594 & 11.64 & 49.50\% \\ \hline
average &   &        &       &  &        &       & 53.88\% \\ \botrule
\end{tabular}
\end{table}

\begin{table}[h]
\caption{Solving performances of CPLEX and EDHR on WC}
\centering
\scriptsize
\begin{tabular}{cccccccc}
\toprule
\multirow{2}{*}{instance} & \multirow{2}{*}{n} & \multicolumn{2}{c}{CPLEX} &  & \multicolumn{2}{c}{EDHR} & \multirow{2}{*}{rate} \\ \cmidrule{3-4} \cmidrule{6-7}
                          &                    & result       & ticks      &  & result      & ticks      &                       \\ \hline
WC01    &  200 &  70447 &  5.94 &  &  70447 &  3.82 & 35.69\% \\
WC02    &  400 & 140429 & 12.29 &  & 140429 & 11.10 &  9.68\% \\
WC03    &  600 & 209320 & 18.45 &  & 209320 & 10.73 & 41.84\% \\
WC04    &  800 & 275931 & 15.50 &  & 275931 & 14.60 &  5.81\% \\
WC05    & 1000 & 346576 & 13.88 &  & 346576 & 10.81 & 22.12\% \\
WC06    & 1200 & 416613 & 21.07 &  & 416613 & 12.63 & 40.06\% \\
WC07    & 1400 & 493611 & 24.10 &  & 493611 & 13.88 & 42.41\% \\
WC08    & 1600 & 548660 & 28.58 &  & 548660 & 16.69 & 41.60\% \\
WC09    & 1800 & 635572 & 29.31 &  & 635572 & 28.49 &  2.80\% \\
WC10    & 2000 & 703742 & 32.77 &  & 703742 & 15.61 & 52.36\% \\ \hline
average &      &        &       &  &        &       & 29.44\% \\ \botrule
\end{tabular}
\end{table}

\begin{table}[h]
\caption{Solving performances of CPLEX and EDHR on SC}
\centering
\scriptsize
\begin{tabular}{cccccccc}
\toprule
\multirow{2}{*}{instance} & \multirow{2}{*}{n} & \multicolumn{2}{c}{CPLEX} &  & \multicolumn{2}{c}{EDHR} & \multirow{2}{*}{rate} \\ \cmidrule{3-4} \cmidrule{6-7}
                          &                    & result       & ticks      &  & result      & ticks      &                       \\ \hline
SC01    &  200 &  33723 & 1.56 &  &  33723 & 0.77 & 50.64\% \\
SC02    &  400 &  69201 & 3.22 &  &  69201 & 1.35 & 58.07\% \\
SC03    &  600 & 103186 & 4.76 &  & 103186 & 2.06 & 56.72\% \\
SC04    &  800 & 137153 & 5.36 &  & 137153 & 2.77 & 48.32\% \\
SC05    & 1000 & 173918 & 5.90 &  & 173918 & 2.75 & 53.39\% \\
SC06    & 1200 & 207026 & 6.34 &  & 207026 & 2.76 & 56.47\% \\
SC07    & 1400 & 239483 & 6.80 &  & 239483 & 3.45 & 49.26\% \\
SC08    & 1600 & 277211 & 7.21 &  & 277211 & 3.81 & 47.16\% \\
SC09    & 1800 & 312242 & 7.44 &  & 312242 & 4.10 & 44.89\% \\
SC10    & 2000 & 343964 & 7.42 &  & 343964 & 3.90 & 47.44\% \\ \hline
average &      &        &      &  &        &      & 51.24\% \\ \botrule
\end{tabular}
\end{table}

\begin{table}[h]
\caption{Solving performances of CPLEX and EDHR on IC}
\centering
\scriptsize
\begin{tabular}{cccccccc}
\toprule
\multirow{2}{*}{instance} & \multirow{2}{*}{n} & \multicolumn{2}{c}{CPLEX} &  & \multicolumn{2}{c}{EDHR} & \multirow{2}{*}{rate} \\ \cmidrule{3-4} \cmidrule{6-7}
                          &                    & result       & ticks      &  & result      & ticks      &                       \\ \hline
IC01    &  200 &  78291 & 1.54 &  &  78291 & 1.54 & 0.00\% \\
IC02    &  400 & 145196 & 3.35 &  & 145196 & 3.35 & 0.00\% \\
IC03    &  600 & 223254 & 3.73 &  & 223254 & 3.73 & 0.00\% \\
IC04    &  800 & 297910 & 4.10 &  & 297910 & 4.10 & 0.00\% \\
IC05    & 1000 & 378112 & 5.37 &  & 378112 & 5.37 & 0.00\% \\
IC06    & 1200 & 459241 & 4.94 &  & 459241 & 4.94 & 0.00\% \\
IC07    & 1400 & 525240 & 6.11 &  & 525240 & 6.11 & 0.00\% \\
IC08    & 1600 & 595483 & 5.49 &  & 595483 & 5.49 & 0.00\% \\
IC09    & 1800 & 675851 & 6.95 &  & 675851 & 6.95 & 0.00\% \\
IC10    & 2000 & 745388 & 7.11 &  & 745388 & 7.11 & 0.00\% \\ \hline
average &      &        &      &  &        &      & 0.00\% \\ \botrule
\end{tabular}
\end{table}

\begin{table}[h]
\caption{Solving performances of CPLEX and EDHR on ASC}
\centering
\scriptsize
\begin{tabular}{cccccccc}
\toprule
\multirow{2}{*}{instance} & \multirow{2}{*}{n} & \multicolumn{2}{c}{CPLEX} &  & \multicolumn{2}{c}{EDHR} & \multirow{2}{*}{rate} \\ \cmidrule{3-4} \cmidrule{6-7}
                          &                    & result       & ticks      &  & result      & ticks      &                       \\ \hline
ASC01   &  200 &  24107 &  3.26 &  &  24107 &  2.35  & 27.91\% \\
ASC02   &  400 &  49283 &  6.53 &  &  49283 &  3.70  & 43.34\% \\
ASC03   &  600 &  71575 &  6.92 &  &  71575 &  3.86  & 44.22\% \\
ASC04   &  800 & 109271 & 17.30 &  & 109271 & 13.49  & 22.02\% \\
ASC05   & 1000 & 125004 & 14.23 &  & 125004 &  7.66  & 46.17\% \\
ASC06   & 1200 & 139315 & 19.38 &  & 139315 &  9.31  & 51.96\% \\
ASC07   & 1400 & 159856 & 19.71 &  & 159856 & 11.28  & 42.77\% \\
ASC08   & 1600 & 189148 & 31.76 &  & 189148 & 26.60  & 16.25\% \\
ASC09   & 1800 & 216785 & 26.28 &  & 216785 & 13.67  & 47.98\% \\
ASC10   & 2000 & 252810 & 30.24 &  & 252810 & 18.78  & 37.90\% \\ \hline
average &      &        &       &  &        &        & 38.05\% \\ \botrule
\end{tabular}
\end{table}

\begin{table}[h]
\caption{Solving performances of CPLEX and EDHR on literature \cite{Jooken2021}}
\centering
\scriptsize
\begin{tabular}{lccccccc}
\toprule
\multirow{2}{*}{instances} & \multirow{2}{*}{n} & \multirow{2}{*}{$|N_{2,1}\cup N_{2,4}|$} & \multicolumn{2}{c}{CPLEX} &  & \multicolumn{2}{c}{EDHR} \\ \cmidrule{4-5} \cmidrule{7-8}
                           &                    &                    & result       & ticks      &  & result      & ticks      \\ \hline
n\_400\_c\_1000000\_g\_2\_f\_0.1\_eps\_0.0001\_s\_200   &  400 &  0 &  505215 &  6.87 &  &  505215 &  6.87 \\
n\_400\_c\_1000000\_g\_2\_f\_0.1\_eps\_0.001\_s\_200    &  400 &  0 &  504648 &  6.85 &  &  504648 &  6.85 \\
n\_400\_c\_1000000\_g\_2\_f\_0.1\_eps\_0.01\_s\_200     &  400 &  0 &  513600 &  7.01 &  &  513600 &  7.01 \\
n\_400\_c\_1000000\_g\_2\_f\_0.1\_eps\_0.1\_s\_200      &  400 &  0 &  604798 &  6.98 &  &  604798 &  6.98 \\
n\_400\_c\_1000000\_g\_2\_f\_0.1\_eps\_0\_s\_200        &  400 &  0 &  503787 &  7.07 &  &  503787 &  7.08 \\ \botrule
\end{tabular}
\end{table}

All computed results for the five types of instances are summarized in Tables 1-5. In these tables, columns 1 and 2 list the names and item scales of the tested instances. Columns 3 and 4 present the results and ticks computed from CPLEX, respectively. Columns 5 and 6 show the results and ticks computed from EDHR, respectively. The final column displays the improvement rate of EDHR relative to CPLEX in terms of ticks.

Table 1-5 shows that EDHR significantly improves problem-solving speed compared to CPLEX, particularly in UC and SC instances, where ticks decreased by an average of 53.88\% and 51.24\%, respectively. In WC instances, EDHR's performance is comparatively weak across individual instances such as WC02, WC04, and WC09, with improvement effects all below 10\%. Nevertheless, the overall effect is considerable, showing an average reduction of 29.44\% in ticks. For ASC instances, EDHR exhibited a relatively consistent improvement, with no instances showing less than a 10\% effect, and an overall average reduction in ticks of 38.05\%.

It is particularly noteworthy that in IC instances, since the sets ($N_{1,1}, N_{2,1}, N_{1,4}$, and $N_{2,4}$) in these instances are empty, EDHR cannot improve the performance for these cases. To further validate this characteristic of EDHR, where it fails to enhance computing speed when the aforementioned four sets are empty, we examined five additional instances from the literature \cite{Jooken2021}. The specific computational results are presented in Table 6.

In Table 6, the displayed content aligns with that of Tables 1-5, with the exception of column 3, which presents the number of items in $N_{2,1}$ and $N_{2,4}$. The table clearly shows that due to $|N_{2,1} \cup N_{2,4}| = 0$, EDHR is unable to decrease the ticks for those instances. Nevertheless, the table also reveals that EDHR does not appreciably hinder the solving speed, thus enabling the algorithm to operate effectively.

\section{Conclusions}

In this paper, we improve Dembo and Hammer's reduction algorithm for the 0-1 Knapsack Problem (0-1 KP) and introduced an extension of their algorithm, referred to as EDHR (Extension Dembo and Hammer's Reduction). Computational results for various instances demonstrate that EDHR significantly improves the solving speed on the majority of 0-1 KP instances, with a pronounced effect.

Our future research will be addressed as the following two issues. First, whether EDHR can combined with other reduction strategy to accelerate the solving speed. Second, we can apply the definition of the set $N_{i,1}$ in EDHR to reduce the search space. By reducing the search space of feasible solution regions that do not contain optimal solutions, the meta-heuristic algorithm not only accelerates the solving of NP-hard problems but also helps to restrict or estimate the parameters of certain operators.

\backmatter


\begin{thebibliography}{15}

\bibitem{Kellerer2004} Hans Kellerer, Ulrich Pferschy, David Pisinger. Knapsack Problems. Springer-Verlag: Berlin Heidelberg, 2004: 1-548.

\bibitem{Pisinger1997} David Pisinger. A Minimal Algorithm for the 0-1 Knapsack Problem. Operations Research, 1997, 45(5): 758-767.

\bibitem{Balas1980} Egon Balas, Eitan Zemel. An Algorithm for Large Zero-One Knapsack Problems. Operations Research, 1980, 28(5): 1130-1154.

\bibitem{Bellman1957} Richard Bellman. Dynamic programming. Princeton University Press, Princeton, 1957: 1-342.

\bibitem{Martello1999} Silvano Martello, David Pisinger, Paolo Toth. Dynamic Programming and Strong Bounds for the 0-1 Knapsack Problem. Management Science,1999,45(3): 414-424.

\bibitem{Martello2000} Silvano Martello,David Pisinger,Paolo Toth. New trends in exact algorithms for the 0-1 knapsack problem. European Journal of Operational Research, 2000, 123(2): 325-332.

\bibitem{Dembo1980} Ron S. Dembo, Peter Ladislaw Hammer. A reduction algorithm for knapsack problems. Methods of Operations Research, 1980, 36(1): 49-60.

\bibitem{Pisinger1995a} David Pisinger. An expanding-core algorithm for the exact 0-1 knapsack problem. European Journal of Operational Research, 1995, 87(1): 175-187.

\bibitem{Garey1979} Michael Randolph Garey, David Stifler Johnson. Computer and Intractablility: A Guide to the Theory of NP-Completeness, Freeman, San Francisco, CA, 1979: 1-338.

\bibitem{Karp1972} Richard Manning Karp. Reducibility among Combinatorial Problems. In: Miller R.E., Thatcher J.W., Bohlinger J.D. (eds) Complexity of Computer Computations. The IBM Research Symposia Series. Springer, Boston, MA, 1972.

\bibitem{Pisinger2016a} David Pisinger, Alima Saidi. Tolerance analysis for 0-1 knapsack problems. European Journal of Operational Research, 2016, 258: 866-876.

\bibitem{Pisinger2009a} Jens Egeblad, David Pisinger. Heuristic approaches for the two- and three-dimensional knapsack packing problem. Computers \& Operations Research, 2009, 36: 1026-1049.

\bibitem{Tsesmetzis2008a} Dimitrios Tsesmetzis, Ioanna Roussaki, Efstathios Sykas. QoS-aware service evaluation and selection. European Journal of Operational Research, 2008, 191: 1101-1112.

\bibitem{Ingargiola1973} Giorgio P. Ingargiola, James F. Korsh. Reduction Algorithm for Zero-One Single Knapsack Problems. Management Science, 1973, 20: 460-463.

\bibitem{Martello1988} Silvano Martello, Paolo Toth. A New Algorithm for the 0-1 Knapsack Problem. Management Science, 1988, 34(5): 633-644.

\bibitem{Martello1990} Silvano Martello, Paolo Toth. Knapsack Problems: Algorithms and Computer Implementations. Wiley, Chichester, UK, 1990.

\bibitem{Fahle2002} Torsten Fahle, Meinolf Sellmann. Cost Based Filtering for the Constrained Knapsack Problem. Annals of Operations Research, 2002, 115(1-4): 73-93.

\bibitem{Dantzig1957} George B. Dantzig. Discrete-Variable Extremum Problems. Operations Research, 1957, 5(2): 266-288.

\bibitem{Dey2023} Santanu S. Dey, Yatharth Dubey, Marco Molinaro. Branch-and-bound solves random binary IPs in poly(n)-time. Mathematical Programming, 2023, 200: 569-587.

\bibitem{Pisinger2005} David Pisinger. Where are the hard knapsack problems ? Computers \& Operations Research, 2005, 32: 2271-2284.

\bibitem{yalmip} L{\"{o}}fberg Johan. YALMIP: A Toolbox for Modeling and Optimization in MATLAB. In \emph{Proceedings of the CACSD Conference}, Taipei, China, 2004.

\bibitem{Wei2019} Wei Zequn, Hao Jin-Kao. Iterated two-phase local search for the Set-Union Knapsack Problem. Future Generation Computer Systems, 2019, 101: 1005-1017.

\bibitem{Jooken2021} Jorik Jooken, Pieter Leyman, Patrick De Causmaecker. A New Class of Hard Problem Instances for the 0-1 Knapsack Problem. European Journal of Operational Research, 2022, 301(3): 841-854.

\end{thebibliography}
\end{document}